\documentclass[a4paper]{article}

\usepackage{amsmath}
\usepackage{amsthm,amscd,amssymb,eucal,mathrsfs}
\usepackage{amsfonts}
\usepackage{latexsym}
\usepackage{dsfont}
\usepackage{url}
\usepackage{listings}
\usepackage[all]{xy}
\usepackage{graphics}

\newtheorem{thm}{Theorem}[section]

\newtheorem{lem}[thm]{Lemma}
\newtheorem{cor}[thm]{Corollary}
\newtheorem{rem}[thm]{Remark}
\newtheorem{dfn}[thm]{Definition}

\newtheorem{alg}[thm]{Algorithm}

\newcommand{\mathset}[1]{{\left\{#1\right\}}} 
\newcommand{\absolute}[1]{\left\lvert#1\right\rvert}

\DeclareMathOperator{\GL}{GL}
\DeclareMathOperator{\id}{id}

\DeclareMathOperator{\supp}{supp}
\DeclareMathOperator{\im}{im}

\DeclareMathOperator{\bin}{bit}
\DeclareMathOperator{\gc}{gc}
\DeclareMathOperator{\coord}{c}

\DeclareMathOperator{\trail}{trail}

\DeclareMathOperator{\static}{static}
\DeclareMathOperator{\scaled}{scaled}

\begin{document}

\title{$p$-Adic scaled space filling curve indices for high dimensional data}
\author{Patrick Erik Bradley, Markus Wilhelm Jahn}

%\affiliation{$^{1}$Institute of Photogrammetry and Remote Sensing, $^{2}$Geodetic Institute; Karlsruhe Institute of Technology, Englerstr.\ 7, 76131 Karlsruhe, Germany}
%\email{$\mathset{\text{bradley,markus.jahn}}$@kit.edu}

%\shortauthors{P.E.\ Bradley and M.W.\ Jahn}

%\received{00 Month 2004}
%\revised{00 Month 2004}

%\keywords{Gray code, Hilbert curve, index, scalability, high dimension, $p$-adic number, non-Archimedean, ultrametric}

\maketitle

\begin{abstract}
Space filling curves are widely used in Computer Science. In particular Hilbert curves and their generalisations to higher dimension are used as an indexing method because of their nice locality properties. This article generalises this concept to the systematic construction of $p$-adic versions of  Hilbert curves based on affine transformations of the $p$-adic Gray code, and develops an efficient scaled indexing method for data taken from high-dimensional spaces based on these new curves, which with increasing dimension is shown to be less space consuming than the optimal standard static Hilbert curve index. A measure is derived which allows to assess the local sparsity of a data set, and is tested on some data.
\end{abstract}

%%%%%%%%%%%%%%%%%%%%%%%%%%%%%%%%%%%%%%%%%%%
%\ToDo{
%- prior to submission, upload \cite{JB-treedist} to arXiv and update bib-file
%}

\section{Introduction}

Motivated by the fact proven by Cantor  that the unit interval has the same cardinality as any Euclidean $n$-space \cite{Cantor1878}, Peano discoverd in 1890 a continuous map from the unit interval onto the unit square, i.e.\ a first example of a space-filling curve \cite{Peano1890}.
Its construction uses an approximation by iteratively subdividing the square into $3^k$ subsquares and aligning these.
A year later, Hilbert modified this $3$-adic approach in order to obtain a binary construction of approximations to  such a curve \cite{Hilbert1891}. These are so-called \emph{space-filling curves} and can be generalised to higher dimension, i.e.\ a continuous map onto an $n$-dimensional hypercube.

\smallskip
The binary reflected Gray code allows for a generalisation of the Hilbert curve to higher dimension.
Such higher dimensional Hilbert curves are used as an index for point clouds in $\mathds{R}^n$.
They have the property that nearby points on the curve  are also nearby in $\mathds{R}^n$.
Another property of Hilbert curves is their fractal structure: The $k+1$-st iteration of the Hilbert curve is the union of copies of the $k$-th iteration.
In \cite{Haverkort2017}, it is asked how many different generalisations of the Hilbert curve are there to dimension $n\ge 3$.
A Gray-code index based on the Gray code for point clouds whose coordinates have varying precision is introduced in \cite{HR2007}.

\smallskip
This article has a similar aim as \cite{HR2007} in that the scalability of the Hilbert curve should be exploited to produce a scalable index whose local precision adapts to the local density
of the point cloud, and to be able to do this for any $p$-adic approximation to a Hilbert-like space-filling curve. For this, the reflected $p$-adic Gray code and affine transformations will be used in order to construct a  multitude of $p$-adic Gray-Hilbert curves.
The emphasis on this article is to develop the theory underlying the $p$-adic Gray-Hilbert curve and to substantiate this theory by some experiental results with small data sets. A bigger data application is undertaken in \cite{JB-treedist} for a large rain forest tree data set.

\bigskip
In the following section, we fix some notation. Section 3 recalls properties of the $p$-adic reflected Gray code. Section 4 deals with the binary reflected code separartely.  Section 5 studies affine transformations of Gray codes.
These are used in Section 6 to construct affine Gray-Hilbert codes. 
 Section 7 goes to the projective limit, where the affine Gray-Hilbert curve becomes a continuous map between formal power series rings. Section 8 shows how to approximate a point on the $2$-adic affine Gray-Hilbert curve, the case $p>2$ having already been done in Section 6. Section 9 introduces a scaled indexing algorithm by mapping points onto different iterations of the swapping Hilbert curve. An insertion algorithm is introduced as well. 
Section 11 deals with upper complexity bounds. This is followed by Section 12 containing experiments. A conclusion finishes the core of this article.

%%%%%%%%%%%%%%%%%%%%%%%%%%%%%%%%%%%%%%%%%%%%%%%%%%%%%%55
\section{Notation}
Let $p$ be a prime number.
The finite field with $p$ elements $0$, $1$, \dots, $p-1$ is denoted as $\mathds{F}_p$. The $n$-dimensional
$\mathds{F}_p$ vector space $\mathds{F}_p^n$ is viewed here as the set of maps
\[
x=(x_{n-1},\dots,x_0)\colon\mathset{0,\dots,n-1}\to\mathds{F}_p,\;
i\mapsto x_i
\]
with the usual point-wise addition and scalar multiplication.
The elements of $\mathds{F}_p^n$ are also coefficients of polynomials or formal power series:
the set of polynomials is denoted as $\mathds{F}_p^n[t]$, and the set of formal power series as $\mathds{F}_p[[t]]$. For $x\in\mathds{F}_p^n$ we denote the support of $x$ as
\[
\supp(x)=\mathset{i\in\mathset{0,\dots,n-1}\mid x_i=1}
\]
The standard unit vectors are denoted as $e_i$  and have 
\[
\supp(e_i)=\mathset{i}
\]
for  $i\in\mathset{0,\dots,n-1}$.
The {Hamming distance} on $\mathds{F}_p^n$ is denoted as
$d_H$.
$S_n$ will denote the group of permutations of the set $\mathset{0,\dots,n-1}$.
If $x\in\mathds{F}_p^r$, then
\[
x^\sigma:=(x_{\sigma(r-1)},\dots,x_{\sigma(0)})
\]
for $x=(x_{r-1},\dots,x_{0})\in\mathds{F}_p^r$.

%%%%%%%%%%%%%%%%%%%%%%%%%%%%%%%%%%%%%%%%%%%%%%%%

\section{$p$-Adic reflected Gray code}

We take the recursive definition of the reflected $p$-adic Gray code from \cite{ZS2003}:
\begin{align*}
  G(1,p)&=(i)_i\\
  G(n,p)&=\left(i\;G(n-1,p)^{\sigma_i}\right)_i
\end{align*}
where $i$ runs through the set $\mathset{0,\dots,p-1}$ and
$\sigma_i\in S_p$ is the permutation
\[
\sigma_i=\begin{cases}
\id,&i\equiv 0\mod 2\\
\sigma&i\equiv 1\mod 2
\end{cases}
\]
where
\[
\sigma=
\begin{pmatrix}
  0&1&\dots&p-1\\
  p-1&p-2&\dots&0
\end{pmatrix}
\]
In other words, $G(n-1,p)^{\sigma_i}$ is $G(n-1,p)$ in reverse order, if $i$ is odd. Notice that $G(n,p)$ is a sequence of elements from $\mathds{F}_p^n$, i.e.\
a map $\mathds{F}_p^n\to\mathds{F}_p^n$.

\smallskip
Let
  \[
\bin_n\colon\mathds{Z}/p^n\mathds{Z}\to\mathds{F}_p^n
\]
be the map which takes every number to its $p$-adic representation.

\smallskip
Instead of $G(n,p)$, we will write $\gc_n$, suppressing the
$p$ in the notation.

\begin{lem}
  It holds true that
  \begin{align*}
    \gc_n(\bin_n(0))&=(0\dots0)\\
    \gc_n(\bin_n(p^n-1))&=(-1\dots-1)
  \end{align*}  
\end{lem}

\begin{proof}
  We prove the assertions by induction.

  \smallskip
  Clearly, we have for $n=1$:
\begin{align*}
  \gc_1(0) &= (0)\\
  \gc_1(p-1)&=(-1)
\end{align*}
Now, assume that the assertions are true for $n>0$. Then
\begin{align*}
  \gc_{n+1}(\bin_{n+1}(0))&=(0\; \gc_n(\bin_n(0))^{\sigma_0})
  \\
  &=(0\; 0\dots 0)\\
  \gc_{n+1}(\bin_{n+1}(p^{n+1}-1))
  &  = (-1 \; \gc_n(\bin_n(p^n-1))^{\sigma_{p-1}})
  \\
&  =(-1\;-1\dots -1)
\end{align*}
This proves the assertions.
\end{proof}

  We say that elements $v$ and $v^\perp$ of $\mathds{F}_p^n$
  are \emph{opposite}, if
  \[
v+v^\perp=d
  \]
with $d=(-1\dots-1)\in\mathds{F}_p^n$. 

\begin{lem}\label{sigmaperp}
  It holds true that
  \[
\gc_n(x)^\sigma=\gc_n(x^\perp)
  \]
\end{lem}

\begin{proof}
  The Gray code in reverse order at $x=\bin_n(i)$ is obtained by
  taking the element of the Gray code at the point whose index is $p^n-1-i$.
  But the $p$-adic representation of this number is
  $d-x=x^\perp$ (observe that in this case, there is no ``carry'' taking place, as the digits of $d$ are all highest possible: $p-1$).
\end{proof}

%%%%%%%%%%%%%%%%%%%%%%%%%%%%%%%%%%%%%%%%%%%%%%%%%%%%%%%%%%%%

\section{Binary reflected Gray code}
Here, we consider the case $p=2$ separately.
Let 
\[
\bin_n\colon\mathds{Z}/2^n\mathds{Z}\to\mathds{F}_2^n %,\;i\mapsto\text{binary representation of $i$}
\]
be the bijective map which takes every number to its binary representation.

\smallskip
A \emph{(cyclic) binary Gray code} is a bijective mapping
\[
g\colon\mathds{F}_2^n\to\mathds{F}_2^n
\]
such that for all $i\in\mathds{Z}/2^n\mathds{Z}$ it holds true that 
\[
d_H(\alpha(i),\alpha(i+1))=1
\]
for $\alpha=g\circ\bin_n$ and
where $d_H$ is the Hamming distance on $\mathds{F}_2^n$.

\smallskip
The \emph{binary reflected Gray code} \cite{Gray1953} is the mapping
\[
\gc_n\colon\mathds{F}_2^n\to\mathds{F}_2^n,\;x\mapsto x+(x\triangleright 1)
\]
where $x\triangleright k$ means a right shift by $k$ bits (cf.\ \cite{Hamilton2006}).

\smallskip
The map $\gamma_n\colon \mathds{Z}/2^n\mathds{Z}\to\mathds{F}_2^n$ is defined as $\gamma_n=\gc_n\circ\bin_n$.
According to \cite[Thmn.\ 2.1]{Hamilton2006}, these definitions are equivalent with the binary reflected Gray code defined there. Notice, that \cite{Hamilton2006} does not distinguish between $\gamma_n$ and $\gc_n$.
For this reason, we also call $\gamma_n$ a \emph{binary Gray code}.

\smallskip
The following observation explains the coinciding notation for $p=2$.

\begin{lem}
The binary reflected Gray code coincides with the $2$-adic reflected Gray code.
\end{lem}

\begin{proof}
This is proven in \cite{Knuth2004}.
\end{proof}

\begin{lem}\label{gc=aut}
The mapping $\gc_n$ is a linear automorphism of the $\mathds{F}_2$-vector space $\mathds{F}_2^n$.
\end{lem}

\begin{proof}
On $\mathds{F}_2^n$ we consider the bitwise addition $+$ of binary numbers. Clearly, $\gc_n$ is a bijection.
It follows from the definition that
\begin{align*}
\gc_n(x)+\gc_n(y)&=x +(x\triangleright 1) + y+ (y\triangleright 1)
\\
&=(x+y)+[(x+y)\triangleright 1]\\
&=\gc_n(x+y)
\end{align*}
as, clearly,
\[
(x\triangleright 1)+(y\triangleright 1)=(x+y)\triangleright 1
\]
Hence, $\gc_n$ is also $\mathds{F}_2$-linear.
\end{proof}

\begin{rem}
Observe that for $p>2$, the Gray code $\gc_n$ is not linear. E.g.\
let $x=(0,\dots,0,1,0)$, $y=(0,\dots,0,p-1,0)$. Then
\begin{align*}
\gc_n(x)&=(0,\dots,0,1,p-1)\\
\gc_n(y)&=(0,\dots,0,p-1,0)\\
\gc_n(x+y)&=\gc_n(\bin_n(0))=\bin_n(0)\\
\gc_n(x)+\gc_n(y)&=(0,\dots,0,p-1)
\end{align*}
which contradicts linearity.
\end{rem}

%%%%%%%%%%%%%%%%%%%%%%%%%%%%%%%%%%%%%%%%%%%%%%%%%%%%%%%%%%%%%%

\section{Transformed Gray code}\label{transformedGC}

\subsection{Case $p=2$}

The binary reflected Gray code $\gamma_n$ has the starting point $e'=(0,\dots,0)$ and the
ending point $f'=(1,0,\dots,0)=e_{n-1}$. The \emph{direction} $d'$ of $\gc_n$ is given by
\[
e'+f'=e_{d'}
\]
So, $d'=n-1$. Any bijection $T$ of $\mathds{F}_2^n$
 which makes $\gc_n$ to a (cyclic) binary Gray code 
 $\gc_T=T^{-1}\circ\gc_n$ must satisfy:
 \[
 d_H(\gamma(i),\gamma(i+1))=1
 \]
 for all $i\in\mathds{Z}/n\mathds{Z}$, where
 \[
 \gamma_T=T^{-1}\circ\gamma_n
 \]
 In particular, $\gc_T$ must satisfy
 \begin{align*}
\gc_T(e')+\gc_T(f')=e_d 
 \end{align*}
 for some $d\in\mathset{0,\dots,n-1}$.
 This means that $T$ takes $\gc_T(e')$ to $e'$ and $e_d$ to $e_{d'}$.
 Furthermore, for all $i\in\mathds{Z}/n\mathds{Z}$ we have
 \[
 \gamma_T(i)+\gamma_T(i+1)=e_j
 \]
and also 
 \[
 \gamma_n(i)+\gamma_n(i+1)=e_{\sigma^{-1}}(j)
\]
with some $j\in\mathds{Z}/n\mathds{Z}$, and a 
permutation $\sigma=\sigma_d$ of $\mathset{0,\dots,n-1}$ induced by $T$,
which takes $n-1$ to $d$.

\smallskip
We will not investigate all possible Hamming distance preserving bijections, but are interested in affine transformations, because they contain a translation which takes $\gc_T(e')$ to $e'$.

\smallskip
Let $e\in\mathds{F}_2^n$. Then the translation with $e$  is denoted as
\[
A_e\colon\mathds{F}_2^n\to\mathds{F}_2^n,\;x\mapsto x+e
\]

\begin{thm}\label{affinePerm}
If $T$ is an affine transformation
such that $\gc_T$ is a Gray code, then it can be written a translation followed by an index permutation:
\[
T=\sigma_d\circ A_e
\]
with prescribed $e\in\mathds{F}_2^n$ %and $d\in\mathset{0,\dots,n-1}$
for some permutation $\sigma_d\in S_n$. % which takes some $d$ to $n-1$.
\end{thm}

Here, we keep $d=\sigma_d^{-1}(n-1)$ in the notation of $\sigma_d$.

\begin{proof}
If $T$ is an affine transformation of $\mathds{F}_2^n$, then it can be written in the form
\[
T\colon x\mapsto \Lambda x + e^\sigma
\]
with some $\Lambda\in\GL_n(\mathds{F}_2)$ and some $e^\sigma\in\mathds{F}_2^n$. We have already seen that $\Lambda$ takes each $e_i$ to some $e_{\sigma(i)}$, and $e_d$ to $e_{n-1}$. Hence,
\[
x\mapsto x^\sigma+e^\sigma=(x+e)^\sigma
\]
This proves the assertion.
\end{proof}

We call $\gc_T$ an \emph{affine Gray code} if 
\[
\gc_T=T^{-1}\circ\gc_n
\]
where $T$ is an affine transformation.

\begin{cor}\label{numberOfGC}
There are precisely $(n-1)!$ affine Gray codes for $\mathds{F}_2^n$ with prescribed starting point $e$ and direction $d$.
\end{cor}

\begin{proof}
This is an immediate consequence of Theorem \ref{affinePerm}.
\end{proof}

%%%%%%%%%%%%%%%%%%%%%%%%%%%%%%%%%%%5
\subsection{Case $p>2$}

  We define a \emph{corner} of $\mathds{F}_p^n$ to be a vector
  $v\in\mathds{F}_p^n$ whose components $v_i$ are in $\mathset{0,-1}$.
The set of all corners of $\mathds{F}_p^n$ will be denoted as $\mathcal{C}$.

\smallskip
Assume $p$ odd. Then let for $e\in\mathcal{C}$:
\begin{align*}
\mathcal{T}_e
=&\left\{\text{affine transformations $T\colon\mathds{F}_p^n\to\mathds{F}_p^n$
    s.t.} \right. 
\\  
&\quad\text{$T(e)=0$, $T(e^\perp)=d$, $T(\mathcal{C})=\mathcal{C}$, and}
\\
&\quad\text{$x-y=e_i\;\Rightarrow\;T(x)-T(y)=\pm e_{\tau(i)}$ for all }
\\&
\left.\quad\text{$i\in\mathset{0,\dots,n-1}$ and
  some permutation $\tau$.}
\right\}
\end{align*}

\begin{thm}\label{affineGrayTransform}
  Let $T\colon\mathds{F}_p^n\to\mathds{F}_p^n$ be an affine transformation.
  It holds true that
  $T\in\mathcal{T}_e$ if and only if $T$ is given as follows:
  \[
T\colon x\mapsto A(x-e)
  \]
  with $A\in \GL_n(\mathds{F}_p)$ such that
  \[
  Ae_i=\begin{cases}
  -e_{\tau(i)},&i\in\supp(e)\\
  +e_{\tau(i)},&i\in\supp(e^\perp)
  \end{cases}
  \]
for some permutation $\tau$ of the set $\mathset{0,\dots,n-1}$.
\end{thm}

We call $A$ the \emph{linear part} of $T$, and $\tau$ the permutation \emph{associated with $T$}.

\begin{proof}
  The affine transformation can be written as
  \[
T\colon x\mapsto Ax+c
\]
with $c\in\mathds{F}_p^n$.
  
  $\Leftarrow$. If $T$ is as asserted, then $c=-Ae$ and
  \[
  T(e)= Ae+c=0
  \]
  Likewise,
  \begin{align*}
  T(e^\perp)&=Ae^\perp-Ae\\
&  =\sum\limits_{i\in\supp(e^\perp)}-e_{\tau(i)}
  +\sum\limits_{i\in\supp(e)}-e_{\tau(i)}
  \\
&  =d
  \end{align*}
  Let $c'\in\mathcal{C}$, and set $C'=\supp(c')$. Then
  \begin{align*}
    T(c')&=Ac'+c\\
&    =\sum\limits_{i\in\supp(e)\cap C'}e_{\tau(i)}
    -\sum\limits_{i\in\supp(e^\perp)\cap C'}e_{\tau(i)}-Ae\\
    &=-\sum\limits_{i\in\supp(e^\perp)}e_{\tau(i)}-\sum\limits_{i\in\supp(e)\setminus C'}e_{\tau(i)}
    \end{align*}
  This is a corner. Now, assume that $x-y=e_i$. Then
  \[
Tx-Ty=Ax-Ay=Ae_i=\pm e_{\tau(i)}
  \]
Hence, we have shown that $T\in\mathcal{T}_e$.
  
  \smallskip
  $\Rightarrow$.
  Let $T\in\mathcal{T}_e$. As $0\in \mathcal{C}$, it follows that
  \[
c=T(0)\in\mathcal{C}
  \]
  Assume now that $x-y=e_i$ for $x,y\in\mathds{F}_p^n$. Then
  \[
Ae_i=Ax-Ay=Tx-Ty=\pm e_{\tau(i)}
\]
for some permutation $\tau\in S_n$. As $T(e)=0$, we have $Ae=-c$. Hence,
\[
-c=Ae=\sum\limits_{i\in\supp(e)}\alpha_i e_{\tau(i)}
\]
with $\alpha_i\in\mathset{\pm 1}$ for $i\in\supp(e)$. As $c\in\mathcal{C}$, it follows that $\alpha_i=1$ for $i\in\supp(e)$. Hence,
\[
A e_i=-e_{\tau(i)}
\]
if $i\in\supp(e)$. As $T(e^\perp)=d$, it follws that
\[
Ae^\perp=d-c=c^\perp\in\mathcal{C}
\]
Hence,
\[
c^\perp=Ae^\perp=\sum\limits_{i\in\supp(e^\perp)}\alpha_i e_{\tau(i)}
\]
with $\alpha_i\in\mathset{\pm 1}$ for $i\in\supp(e^\perp)$. It follows that
$\alpha_i=-1$ for all $i\in\supp(e^\perp)$. Hence,
\[
A e_i=e_{\tau(i)}
\]
for all $i\in\supp(e^\perp)$. This proves the assertion about $T$.  
\end{proof}

\begin{cor}\label{noOfGrayCodes}
It holds true that $\absolute{\mathcal{T}_e}=n!$.
\end{cor}

\begin{proof}
  $T\in\mathcal{T}_e$ is completely determined by the associated permutation $\tau$
as  in Theorem \ref{affineGrayTransform}.
The map
  \[
\mathcal{T}_e\to S_n,\;T\mapsto\tau
\]
is clearly bijective.
\end{proof}

\begin{cor}\label{A&d}
  Let $A$ be the linear part of $T\in\mathcal{T}_e$. Then
  \[
A(e^\perp-e)=d
  \]
\end{cor}

\begin{proof}
  From Theorem \ref{affineGrayTransform}, we have
  \begin{align*}
  A(e^\perp-e)&=\sum\limits_{i\in\supp(e^\perp)}-Ae_i+\sum\limits_{i\in\supp(e)}Ae_i
  \\
&  =\sum\limits_{i=0}^{n-1}-e_{\tau(i)}=d
  \end{align*}
  as asserted.
\end{proof}

\begin{cor}
The transformation $T\in\mathcal{T}_e$ maps opposite points to opposite points.
\end{cor}

\begin{proof}
  Let $x\in\mathds{F}_p^n$. Then
  \begin{align*}
    T(x)+T(x^\perp)&=A(x-e)+A(x^\perp-e) 
    \\
&    = A(d-2e)
    =A(e^\perp-e)=d
  \end{align*}
  where the last equality holds true by Corollary \ref{A&d}. Hence, $T(x)$ and $T(x^\perp)$ are opposite points.
\end{proof}

\begin{cor}\label{Atauperp}
  Let $T\in\mathcal{T}_e$ with $A\in\GL_n(\mathds{F}_p)$ its linear part and $\tau\in S_n$ its associated permutation as in Theorem \ref{affineGrayTransform}. Then
  \begin{align*}
    Ae&=-e^\tau\\
    Ae^\perp&=\left(e^\tau\right)^\perp
    \end{align*}
\end{cor}

\begin{proof}
  This follows from the definition of $A$, as
  $\left(e^\perp\right)^\tau=\left(e^\tau\right)^\perp$. 
\end{proof}

\begin{cor}\label{Tinverse}
If $T\in\mathcal{T}_e$, then $T^{-1}\in\mathcal{T}_{e^\tau}$, where $\tau\in S_n$ is the permutation associated with $T$.
\end{cor}

\begin{proof}
  Let $A\in\GL_n(\mathds{F}_p)$ be the linear part of $T$. Then $T^{-1}$ is given as
  \[
T^{-1}(x)= A^{-1}x+e=A^{-1}(x+Ae)=A^{-1}(x-e^\tau)
  \]
  according to Corollary \ref{Atauperp}. From this it follows that
  \[
T^{-1}(e^\tau)=0
\]
and
\begin{align*}
T^{-1}\left(\left(e^\tau\right)^\perp\right)
&=A^{-1}\left(\left(e^\perp\right)^\tau-e^\tau\right)
\\
&=A^{-1}(Ae^\perp+Ae)
\\
&=e^\perp+e=d
\end{align*}
Clearly, $T^{-1}$ induces a bijection on $\mathcal{C}$. Further, $A^{-1}$ is given by
\[
Ae_i=\begin{cases}
-e_{\tau^{-1}(i)},&i\in\supp\left(e^\tau\right)\\
e_{\tau^{-1}(i)},&i\in\supp\left(\left(e^\tau\right)^\perp\right)
\end{cases}
\]
This means that if $x-y=e_i$, then
\[
T^{-1}(x)-T^{-1}(y)=A^{-1}(x-y)=A^{-1}e_i=\pm e_{\tau^{-1}(i)}
\]
From these considerations, it follows that $T^{-1}\in\mathcal{T}_{e^\tau}$.
\end{proof}

\begin{cor}\label{Tperp}
  Let $T\in\mathcal{T}_e$. Then
  \[
T(x)^\perp=T\left(x^\perp\right)
\]
for all $x\in\mathds{F}_p^n$.
\end{cor}

\begin{proof}
  It holds true that
  \begin{align*}
    T(x)^\perp&=d-T(x)=T(e^\perp)-T(x)
    \\
&    =A(e^\perp-e)-A(x-e)\\
&    = A(e^\perp-x)=A(x^\perp-e)\\
&    =T(x^\perp)
  \end{align*}
 as asserted. 
\end{proof}

In what follows, assume that $e\in\mathcal{C}$ and $T_e\in\mathcal{T}_e$ are fixed.

\smallskip
Let
\[
\gc_n^e:=T_e^{-1}\circ\gc_n\colon\mathds{F}_p^n\to\mathds{F}_p^n
\]

\begin{lem}
  It holds true that
  \[
    \gc_n^e(0)=e,\quad
    \gc_n^e(d)=e^\perp
  \]
  for $e\in\mathds{F}_p^n$.
\end{lem}

\begin{proof}
  We have
  \begin{align*}
    \gc_n^e(0)&=T_e^{-1}(\gc_n(0))=T_e^{-1}(0)=e\\
    \gc_n^e(d)&=T_e^{-1}(\gc_n(d))=T_e^{-1}(d)=e^\perp
  \end{align*}
where the last equalities hold true  by definition of $T_e$.
\end{proof}

%%%%%%%%%%%%%%%%%%%%%%%%%%%%%%%%%%%%%%%%%%%%%%%%%%%%%%%%%%%%%%5
\section{Affine Gray-Hilbert curves}

In \cite{Hamilton2006} there is a construction of a Hilbert curve in arbitrary dimension $n$ by applying the Gray code to each subhypercube after a certain affine transformation.
Here, instead of this particular one, we allow any affine transformation to obtain a transformed Gray code as in Section \ref{transformedGC}.

\smallskip
The \emph{zero-th iteration} of this Hilbert curve is the $p$-adic reflected Gray code for $\mathds{F}_p^n$.
The \emph{$k+1$-st iteration} is obtained from the $k$-th iteration by subdividing each hypercube into $p^n$ subhypercubes and glueing together transformed Gray code curves for each of those hypercubes.
These Gray code pieces are called \emph{local pieces} and are
determined by the local pieces from the $k$-the iteration.

\begin{dfn}
A curve constructed as above is called an \emph{affine Gray-Hilbert curve}.
\end{dfn}

\begin{thm}
The number of $k$-th iterations of an affine Gray-Hilbert curve equals
\[
\begin{cases}
(n-1)!\cdot\alpha_2,&p=2\\
n!\cdot\alpha_p,&p>2
\end{cases}
\]
with
\[
\alpha_p =\sum\limits_{i=0}^k (p^n)^k=\frac{p^{n(k+1)}-1}{p^n-1}
\]
for all $p$.
\end{thm}

\begin{proof}
This follows from Corollary \ref{numberOfGC} for $p=2$, or 
Corollary \ref{noOfGrayCodes} for $p>2$, by induction.
\end{proof}

There is a natural aggregation map between the $k+1$-st iteration and the $k$-th iteration of an affine Gray-Hilbert curve which takes each hypercube to the superhypercube containing it.

\smallskip
In the following two subsections, we are going to explicitly
construct affine Gray-Hilbert curves.

%%%%%%%%%%%%%%%%%%%%%%%%%%%%%%%%%%%%%%
\subsection{Case $p=2$}
Consider the $i$-th point of the $k$-th iteration of an affine Gray-Hilbert curve for $[0,1]^n$.
In the $k+1$-st iteration, this point is replaced by a transformed Gray code $\gamma_n^{(e,d)}$, 
depending on $e\in\mathds{F}_2^n$ and $d\in\mathds{Z}/n\mathds{Z}$.
We may assume that this point lies in the image under the aggregation map of an untransformed Gray code curve $\gamma_n$. 
Here, the first point of $\gamma_n^{(e,d)}$ 
is $e$ and the direction is  $d$.
Following  Section \ref{transformedGC}, we have
\[
\gamma_n^{(e,d)}= T_{(e,d)}^{-1}\circ \gamma_n
\]
where we choose the particular affine  transformation
\[
T_{(e,d)}=\sigma_d\circ A_e
\]
with $\sigma_d$ a permutation which takes $n-1$ to $d$.

\smallskip
We also define
\[
\gc_n^{(e,d)}:=T_{(e,d)}^{-1}\circ \gc_n
\]

Next,  we get 
from the definition of the binary reflected Gray code that
\[
\gamma_n(i)+\gamma_n(i+1)=e_{\tau_n(i)}
\]
with some $\tau_n(i)\in\mathds{Z}/n\mathds{Z}$
for all $i\in\mathds{Z}/2^n\mathds{Z}$. 
According to \cite[Lemma 2.3]{Hamilton2006}, it holds true that
\[
\tau_n(i)=\text{\# of trailing $1$ bits of $\bin_n(i)$}\mod n
\]

This leads to:

\begin{cor}\label{transformedGray}
It holds true that
\[
\gamma_n^{(e,d)}(i)+\gamma_n^{(e,d)}(i+1)
=e_{\tau_n(i)}^{\sigma_d}
\]
%with $\sigma_d=(d \; n-1)$.
\end{cor}

\begin{proof}
It holds true that
\begin{align*}
\gamma_n^{(e,d)}(i)&+\gamma_n^{(e,d)}(i+1)\\
&=\gamma_n(i)^{\sigma_d}+e+\gamma_n(i+1)^{\sigma_d}+e
\\
&=e_{\tau_n(i)}^{\sigma_d}
\end{align*}
This proves the assertion.
\end{proof}

We define:
\begin{align}\label{tau}
\tau_n^{(e,d)}&\colon\mathds{Z}/2^n\mathds{Z}\to\mathds{Z}/n\mathds{Z},
\\
i&\mapsto\log_2\left(\bin_n^{-1}\left(e_{\tau_n(i)}^{\sigma_d}\right)\right)\mod n\nonumber
\end{align}
Notice that $\tau_n^{(e,d)}$ does not depend on $e$. We keep $e$ in the notation anyway, as it is part of a transformation of parts of Hilbert curves depending on
$e$ and $d$.

\smallskip
We will need the following map:
\begin{align*}
\bin_{n}^k\colon&\mathds{Z}/(2^n)^k\mathds{Z}\to\mathds{F}_2^n[t]/t^k\mathds{F}_2^n[t]
\\
&x=\sum\limits_{\nu=0}^{k-1}\alpha_\nu\cdot(2^n)^\nu\mapsto\sum\limits_{\nu=0}^{k-1}\bin_n(\alpha_\nu)\cdot t^{k-1-\nu}
\end{align*}
where $x$ is given by its $2^n$-adic expansion.

\smallskip
Observe that in the following map:
\begin{align*}
\coord_n^k\colon&\mathds{F}_2^n[t]/t^k\mathds{F}_2^n[t]\to[0,1]^n
\\
&\sum\limits_{\nu=0}^{k-1}a_\nu t^\nu\mapsto\sum\limits_{\nu=1}^{k} a_\nu\left(\frac12\right)^\nu
\end{align*}
that $\coord_n^k(f)$ is the lower left corner of a sub-hypercube of $[0,1]^n$ of sidelength $\left(\frac12\right)^k$.

\smallskip
There is the reduction map
\begin{align*}
\rho_k\colon\mathds{F}_2^n[t]/t^{k+1}\mathds{F}_2^n[t]&\to\mathds{F}_2^n[t]/t^k\mathds{F}_2[t]\\
f&\mapsto f\mod t^k
\end{align*}

Observe that in the commutative diagram
\begin{align}\label{continuous2Cube}
\xymatrix{
\mathds{F}_2^n[t]/t^{k+1}\mathds{F}_2^n[t]\ar[r]^{\quad\coord_n^{k+1}}\ar[d]_{\rho_k}&\im(\coord_n^{k+1})\ar[d]^{\mu_k}\\
\mathds{F}_2^n[t]/t^k\mathds{F}_2^n[t]\ar[r]_{\quad\coord_n^k}&\im(\coord_n^k)
}
\end{align}
the pre-image of a point under $\mu_k$ induces a subdivision of the corresponding hypercube into $2^n$ hypercubes.

\smallskip
The following commutative diagram with bijective horizontal maps
\begin{align}\label{pi}
\xymatrix{
\mathds{Z}/(2^n)^{k+1}\mathds{Z}\ar[r]^{\bin_n^{k+1}}\ar[d]_{\pi_k}&\mathds{F}_2^n[t]/t^{k+1}\mathds{F}_2^n[t]\ar[d]^{\rho_k}\\
\mathds{Z}/(2^n)^k\mathds{Z}\ar[r]_{\bin_n^k}&\mathds{F}_2^n[t]/t^k\mathds{F}_2^n[t]
}
\end{align}
defines the vertical map $\pi_k$. 

\begin{lem}
It holds true that the representatives in $\mathset{0,\dots,2^{n(k+1)}-1}$ of
\[
\pi_k^{-1}(m)
\]
for any $m\in\mathds{Z}/(2^n)^k\mathds{Z}$
form a consecutive set of numbers.
\end{lem}

\begin{proof}
See how (\ref{pi}) looks on elements:
\[
\xymatrix{
\sum\limits_{\nu=0}^k\alpha_\nu(2^n)^\nu\ar@{|->}[r]\ar@{|->}[d]
&\sum\limits_{\nu=0}^k\alpha_\nu t^{k-\nu}\ar@{|->}[d]\\
m\ar@{|->}[r]&\sum\limits_{\nu=1}^k\alpha_\nu t^{k-\nu}
}
\]
As 
\[
\sum\limits_{\nu=1}^k\alpha_\nu t^{k-\nu}=\sum\limits_{\nu=0}^{k-1}\alpha_{\nu+1}t^{k-1-\nu}
\]
it follows that
\[
m=\sum\limits_{\nu=0}^{k-1}\alpha_{\nu+1}(2^n)^\nu
\]
and all possible pre-images are of the form 
\[
\alpha_0+2^n\cdot m
\] 
with $\alpha_0\in\mathset{0,\dots,2^n-1}$.
This forms a set of consecutive numbers.
\end{proof}

Let $b+\mathds{F}_2^n t^k$ be the pre-image  of
 $b\in\mathds{F}_2^n[t]/t^k\mathds{F}_2^n[t]$ under the reduction map
 $\rho_k$. Then we define a bijection
\[
D_t^k\colon b+\mathds{F}_2^n t^k\to\mathds{F}_2^n,\; 
b+\alpha t^k\mapsto \alpha
\]
and  obtain a map $h_b^k$ via the commutative diagram:
\begin{align}\label{localHilbert}
\xymatrix{
b+\mathds{F}_2^n t^k\ar[r]^{h^k_b}\ar[d]_{D_t^k}&b+\mathds{F}_2^n t^k\ar[d]^{D_t^k}
\\
\mathds{F}_2^n\ar[r]_{\gc_n^{(e_k(b),d_k(b))}}&\mathds{F}_2^n
}
\end{align}
with $e_k(b)\in\mathds{F}_2^n$ and $d_k(b)\in\mathds{Z}/n\mathds{Z}$.

\smallskip
We define a distance $d$ on $\mathds{F}_2^n/t^k\mathds{F}$
as
\[
d\left(\sum\limits_{\nu=0}^{k-1}\alpha_\nu t^\nu,
\sum\limits_{\nu=0}^{k-1}\beta_\nu t^\nu\right)=
\sum\limits_{\nu=0}^{k-1}d_H(\alpha_\nu,\beta_\nu)
\]

\begin{thm}\label{glueHilbert}
For all natural $k\in\mathds{N}$ there exist functions
\begin{align*}
e_{k}&\colon\mathds{F}_2^n[t]/t^{k}\mathds{F}_2^n[t]\to\mathds{F}_2^n\\
d_{k}&\colon\mathds{F}_2^n[t]/t^{k}\mathds{F}_2^n[t]\to\mathds{Z}/n\mathds{Z}
\end{align*}
such that the maps $h_b^{k}$ are the restrictions of maps $h^{k+1}$ fitting into the commutative diagram
\[
\xymatrix{
\mathds{F}_2^n[t]/t^{k+1}\mathds{F}_2^n[t]\ar[r]^{h^{k+1}}\ar[d]_{\rho_k}
&\mathds{F}_2^n[t]/t^{k+1}\mathds{F}_2^n[t]\ar[d]^{\rho_k}
\\
\mathds{F}_2^n[t]/t^k\mathds{F}_2^n[t]\ar[r]_{h^{k}}
&\mathds{F}_2^n[t]/t^k\mathds{F}_2^n
}
\]
and such that
with
\begin{align*}
p_\ell&=h^{k+1}\left(\bin_n^{k+1}\left(2^{n}\cdot\ell+(2^{n}-1)\right)\right)\\
q_\ell&=h^{k+1}\left(\bin_n^{k+1}\left(2^n\cdot(\ell+1)\right)\right)
\end{align*}
it holds true that
\[
d(\rho_k(p_\ell),\rho_k(q_\ell))=\delta_{k,0}
\]
for all $\ell\in\mathset{0,\dots,2^{nk}-1}$.
\end{thm}

We call $h_b^k$ a \emph{local piece} of $h^{k+1}$, and
$h^{k+1}$ is called a \emph{refinement} of $h^k$.

\begin{proof}
We prove the assertion by induction: namely that for all $N\in\mathds{N}$
the assertion holds true for all $k\le N$.

\smallskip
$N=0$. In this case, we have to check only the one diagram
\[
\xymatrix{
\mathds{F}_2^n\ar[r]^{h^1}\ar[d]_{\rho_0}&\mathds{F}_2^n\ar[d]^{\rho_0}\\
0\ar[r]_{h_0}&0
}
\]
where $h^0$ is the zero map. 
According to (\ref{localHilbert}), this works with 
$h^1=\gc_n=\gc_n^{(0,0)}$ being the Gray code, i.e. $e_0$ and $d_0$ map everything to zero. The statement about the distance is clearly true.

\smallskip
Assume now that the assertion holds true for all $k\le N-1$ with $N\ge 1$.
This means that we have maps $e_{N-1}$ and $d_{N-1}$ such that the maps
$h_b^{N-1}$ for $b\in\mathds{F}_2^n[t]/t^{N-1}\mathds{F}_2^n[t]$ are the restrictions of a map 
\[
h^{N}\colon\mathds{F}_2^n[t]/t^{N}\mathds{F}_2^n[t]\to
\mathds{F}_2^n[t]/t^{N}\mathds{F}_2^n[t]
\]
to $b+\mathds{F}_2^n t^{N}$ satisfying the distance condition.
Let $h_b^{N-1}$ be a local piece of $h^N$. Without loss of generality, 
we may assume that $h_b^{N-1}=\gc_n$. Then we need to construct a Hilbert curve $\chi_b$
such that
\[
\xymatrix{
\mathds{F}_2^n[t]/t^2\mathds{F}_2^n[t]\ar[r]^{\chi_b}\ar[d]_\rho
&\mathds{F}_2^n[t]/t^2\mathds{F}_2^n[t]\ar[d]^\rho
\\
\mathds{F}_2^n\ar[r]_{\gc_n}&\mathds{F}_2^n
}
\]
where $\rho$ is the canonical projection. Let $a\in\mathds{F}_2^n$.
We need to show that $h_a$ in the following diagram is a local piece of $\chi_b$:
\[
\xymatrix{
\rho^{-1}(a)=a+\mathds{F}_2^n t\ar[r]^{\qquad h_a}\ar[d]_{\lambda_a}
&a+\mathds{F}_2^n t\ar[d]^{\lambda_a}
\\
\mathds{F}_2^n\ar[r]_{\gc_n^{(e_a,d_a)}}&\mathds{F}_2^n
}
\]
Let $i=\gamma_n^{-1}(a)$, and set $e_a=e(i)$ and $d_a=d(i)$ as in \cite{Hamilton2006}, where it is shown that the family of
$h_a$ form a second iteration Hilbert curve $\chi_b$. We now need to show that each $\chi_b$ is a local piece of a Hilbert curve $h^{N+1}$. Strictly speaking, \cite{Hamilton2006} shows only that the second iteration
is a Hilbert curve. 

\smallskip
$h^N$ being a Hilbert curve means that (via $c_n^N$) the local pieces $h_b^{N-1}$ traverse hypercubes $C$ which themselves are elements of a Hilbert curve. Let $b,b'$ be such that $C,C'$ is an ordered pair of neighbouring hypercubes. Thus, the last hypercube traversed by $h_b^{N-1}$ contains a vertex $v$ of $C$, and the first hypercube traversed by $h_{b'}^{N-1}$ contains a vertex $v'$ of $C$. As $h^N$ is a Hilbert curve, it follows that 
\[
v=v'
\]
Now, refinements $\chi_b$ and $\chi_{b'}$ are also Hilbert curves. Hence, the last hypercube traversed by
$\chi_b$ contains a vertex $w$ of $C$,
and the first hypercube traversed by $\chi_{b'}$ contains a vertex $w'$ of  $C'$. As those Hilbert curves are refinements of $h_b^{N-1}$ and $h_{b'}^{N-1}$, respectively, it follows that 
\[
w=v\;\text{and}\; w'=v'
\]
Hence, $v=v'$. This proves that the distance between those two small hypercubes also equals one. Hence, $h^{N+1}$ is a Hilbert curve.
\end{proof}

%\smallskip
%The main reason for this particular choice of Gray code transformation is:

%\begin{lem}
%Let $\sigma=(i\;j)$ be a transpositon with $i\neq j$. Then
%\[
%x^{\sigma}=
%\begin{cases}
%x,&x_i=x_j\\
%A_{e_i+e_j}(x),&x_i\neq x_j
%\end{cases}
%\]
%where $x=(x_{n-1},\dots,x_{0})\in\mathds{F}_2^n$.
%\end{lem}

%\begin{proof}
%If $x_i=x_j$, then the statement is clearly true.
%Assume now that $x_i\neq x_j$. Then
%\[
%x^\sigma=x+e_i+e_j=A_{e_i+e_j}(x)
%\]
%as both $x_i$ and $x_j$ are replaced by $x_i+1$ and $x_j+1$, %respectively.
%This proves the assertion.
%\end{proof}

%Consequently, transpositions $(i\; j)$ act either trivially on $x\in\mathds{F}_2^n$ or as if  by interchanging the branches in $\mathcal{T}_n$ rooted in the vertices of heights $n-1-i$ and $n-1-j$,
%according to Lemma \ref{interchange}.

%\begin{dfn}
%The affine Gray-Hilbert curve whose local pieces are $\gc_n^{(e,d)}$ is called the \emph{swapping Hilbert curve}.
%\end{dfn}

%%%%%%%%%%%%%%%%%%%%%%%%%%%%%%%%%
\subsection{Case $p>2$}
If we view $x\in\mathds{F}_p^n$ as a word over the alphabet $\mathset{0,\dots,p-1}$, then we can define the \emph{longest common prefix length}
as the number $\ell(x,y)$ of common initial consecutive letters of $x$ and $y$.
Then we can define
\[
\trail_{-1}(x):=\ell(x,d)
\]
for $x\in\mathds{F}_p^n$, and $d=(-1\dots-1)$.

\begin{lem}\label{trail}
  Let $k=\trail_{-1}(x)$ for $x=\bin_n(i)$ with
  $i\in\mathds{Z}/p^n\mathds{Z}$. Then
  \[
\bin_n(i+1)=(y_\nu)
\]
with
\[
y_\nu=\begin{cases}
0,&\nu< k\\
x_{k}+1,&\nu=k\\
x_\nu,&\nu>k
\end{cases}
\]
and $\nu\in\mathset{0,\dots,n-1}$.
\end{lem}

\begin{proof}
For $\nu<k$, all $x_\nu$ get raised by one from $p-1$ to $0$, and then there is a last ``carry'', which raises $x_k$ by one. As $x_k\neq p-1$, it follows that there is no ``carry'' beyond this place, and so $y_\nu=x_\nu$ for $\nu>k$.
\end{proof}

Let
\[
\Delta(i):=\gc_n(\bin_n(i+1))-\gc_n(\bin_n(i))
\]

\begin{lem}
  It holds true that
  \[
  \Delta(i)
  =\pm e_{g(i)}
  \]
  is plus or minus a standard unit vector with
   \[
  g(i)=
  \trail_{-1}(\bin_n(i))
  \]
for all $i\in\mathds{Z}/p^n\mathds{Z}\setminus\mathset{-1}$.
\end{lem}

\begin{proof}
The first statement is proven in \cite{Guan1998}.

\smallskip\noindent
Let $x=\bin_n(i)$, $y=\bin_n(i+1)$, $x'=\gc_n(x)$, $y'=\gc_n(y)$, and
$k=\trail_{-1}(x)$. 
From Lemma \ref{trail}, it follows that the last $n-k-1$ digits of
$x'$ and $y'$ coincide. As
\[
x_{k}\neq y_{k}
\]
it follows that 
\[
x_k'=\tau(x_k)\neq\tau(y_k)=y_k'
\]
with either $\tau=\id$ or $\tau=\sigma$ for the permutation $\sigma$ in the definition of the Gray code.
This shows that
\[
g(i)=k
\]
as asserted.
\end{proof}

View the $p$-adic Gray code $G(n,p)$ as an ordering of points on a hypercube
$I^n$.
Assume that this hypercube is further regularly subdivided into $p^n$ hypercubes, each of which is ordered by a transformed $p$-adic Gray code $\gc_n^e$ for some corner $e\in\mathds{F}_p^n$ as follows: 
Let $\epsilon(i)$ and $\phi(i)$ be the entry and exit points of the $i$-th subhypercube in the transformed Gray code ordering of that individual subhypercube.
The $i$-th and the $i+1$-th subhypercubes are neighbors along the $g(i)$-th coordinate, we require
\[
\epsilon(i+1)-\phi(i)=e_{g(i)}
\]
for $i\in\mathds{Z}/p^n{Z}\setminus\mathset{-1}$. We call this construction the
\emph{first iteration of a $p$-adic Gray-Hilbert curve}.

\begin{lem}
Entry and exit points of a given subhypercube are corners and are opposite to each other:
\[
\epsilon(i)\in\mathset{0,-1}^n,\quad \phi(i)=\epsilon(i)^\perp
\]
for $i\in\mathds{Z}/p^n\mathds{Z}\setminus\mathset{-1}$.
\end{lem}

\begin{proof}
  Now, $\epsilon(i)=\gc_n^e(0)$ for some corner $e\in\mathds{F}_p^n$. It follows
  that
  \[
\epsilon(i)=\gc_n^e(0)=T_e^{-1}(\gc_n(0))=T_e^{-1}(0)=e
\]
Thus, $\epsilon(i)$
is a corner. Further, 
\[
\phi(i)=\gc_n^e(d)=T_e^{-1}(\gc_n(d))=T_e^{-1}(d)=e^\perp=\epsilon(i)^\perp
\]
as asserted.
\end{proof}

Hence, we have
\begin{align}\label{entry-rec}
  \epsilon(i+1)-\epsilon(i)^\perp
  =e_{g(i)}
\end{align}
for $i\in\mathds{Z}/p^n\mathds{Z}\setminus\mathset{-1}$.

\smallskip
Further,
$\epsilon(0)$ equals the entry point of the parent hypercube, and
$\phi(p^n-1)$ equals the exit point of the parent hypercube.

\begin{lem}
  The entry points of the first iteration of a $p$-adic Gray-Hilbert curve satisfy:
  \[
  \epsilon(i)=
  \sum\limits_{j=0}^{i-1}(-1)^{i-1-j}(e_{g(j)}+d)
\]
for $i\in\mathset{0,\dots,p^n-1}$.
\end{lem}

\begin{proof}
This follows by induction from (\ref{entry-rec}).
\end{proof}

Let
\[
\epsilon^e:=T_e^{-1}\circ\epsilon\colon\mathds{Z}/p^n\mathds{Z}\to\mathds{F}_p^n
\]
for $e\in\mathds{F}_p^n$.
Let $i\in\mathds{Z}/(p^n)^{k+1}\mathds{Z}$, and let
\[
P_i(t)=\sum\limits_{\lambda=0}^k i_\lambda t^\lambda\in\mathds{Z}/p^n\mathds{Z}[t]
\]
be such that $P_i(p^n)$ is the $p^n$-adic expansion of $i$.
$P_i(t)$ is called the \emph{$p^n$-adic polynomial} for $i$.  

\begin{dfn}
  A \emph{reflected $p$-adic Gray-Hilbert curve} is a map
  \[
  h_n\colon\mathds{F}_p^n[[t]]\to\mathds{F}_p^n[[t]],\;
  \sum\limits_{k\in\mathds{N}}x_kt^k
  \mapsto\sum\limits_{k\in\mathds{N}}\gc_n^{\epsilon_k}(x_k)t^k
  \]
  where  $\epsilon_k=\epsilon_k(x_k)\in\mathds{F}_p^n$ are corners such that
  \[
\epsilon_{k+1}-\epsilon_k^\perp=\pm e_{\tau_k}
\]
for some $\tau_k\in S_{n}$ and all $k\in\mathds{N}$.
\end{dfn}

\begin{thm}\label{p-adicHC}
  Consider for each $i\in\mathds{Z}/(p^n)^{k+1}$ the sequence
  \[
  \epsilon_0=0,\quad
\epsilon_{\lambda+1}=\epsilon^{\epsilon_\lambda}(i_\lambda)
\]
where $i_\lambda\in\mathds{Z}/p^n\mathds{Z}$ is a coefficient of the $p^n$-adic polyonomial
\[
P_i(t)=\sum\limits_{\lambda\in\mathds{N}}i_\lambda t^\lambda
\]
for $i$. Then it holds true that
\[
h_n(x)=\sum\limits_{\lambda\in\mathds{N}}\gc_n^{\epsilon_\lambda}(x_\lambda)t^\lambda
\]
with $x=\sum\limits_{\lambda\in\mathds{N}}x_\lambda t^\lambda$ and
\[
x_\lambda=\bin(i_\lambda)
\]
yields a well-defined $p$-adic Gray-Hilbert curve
\[
h_n\colon\mathds{F}_p^n[[t]]\to\mathds{F}_p^n[[t]],\;x\mapsto h_n(x)
\]
for any $p>2$ prime and $n>0$.
\end{thm}

We call
\begin{align*}
h_n^k&\colon \mathds{F}_p^n[t]/t^{k+1}\mathds{F}_p^n[t]\to
\mathds{F}_p^n[t]/t^{k+1}\mathds{F}_p^n[t],
\\
&x\mapsto
h_n^k(x)=\sum\limits_{\lambda=0}^k\gc_n^{\epsilon_\lambda}(x_\lambda)t^\lambda
\end{align*}
 the \emph{$k$-th iteration} of $h_n$.

\begin{proof}
  We need to show that the Hamming distance between the exit point
  of one curve piece and the entry point 
  of the next curve piece equals one for any iteration $k$ of the $p$-adic Gray-Hilbert curve. Let $h_n^k$ be the $k$-th iteration of $h$. Then two consecutive Gray-code pieces belonging to neighouring hypercubes at level $k$ come from
  the same hypercube at level $k-1$. Hence, we are interested in what happens for points $x=\bin_n(i)$, $x'=\bin_n(i')$ with
  $i'_{k}=i_{k}-1$, and $i'_\lambda=i_\lambda$ for $\lambda < k$.
  This means that $\epsilon_{k}$ coincides for both points.
  The difference between the exit point at $i_{k}-1$ and entry point at $i_k$ is thus
  \begin{align*}
  \epsilon^{\epsilon_{k}}(i_{k}-1)^\perp
  &- \epsilon^{\epsilon_{k}}(i_{k})
  =T_{\epsilon_k}^{-1}(\epsilon(i_k))
  -T_{\epsilon_k}^{-1}(\epsilon(i_k-1))^\perp
  \\
&  \stackrel{\ref{Tinverse} + \ref{Tperp}}{=}
  T_{\epsilon_k}^{-1}(\epsilon(i_k))
  -T_{\epsilon_k}^{-1}(\epsilon(i_k-1)^\perp)
  \\
  &=A_k^{-1}(\epsilon(i_k)-\epsilon(i_k-1))
  \stackrel{(\ref{entry-rec})}{=}A_k^{-1}(e_{g(i_k)})
  \\
  &=\pm e_{\tau_k^{-1}(g(i_k))}
    \end{align*}
  where $A_k^{-1}$ is the linear part of and $\tau_k$ the permutation associated with $T_{\epsilon_k}^{-1}$.
\end{proof}

%%%%%%%%%%%%%%%%%%%%%%%%%%%%%%%%%%%%%%%%%%%%%%%%%%%%%%%%%%55
\section{Taking the projective limits}

The diagram (\ref{pi}) is commutative for all $p$:

\begin{align}
\xymatrix{
\mathds{Z}/(p^n)^{k+1}\mathds{Z}\ar[r]^{\bin_n^{k+1}}\ar[d]_{\pi_k}&\mathds{F}_p^n[t]/t^{k+1}\mathds{F}_p^n[t]\ar[d]^{\rho_k}\\
\mathds{Z}/(p^n)^k\mathds{Z}\ar[r]_{\bin_n^k}&\mathds{F}_p^n[t]/t^k\mathds{F}_p^n[t]
}
\end{align}

where $\rho_k$ is the canonical projection.

\smallskip
The horizontal bijections $\bin_n^k$ induce a homeomorphism between the projective limits:
\[
\xymatrix{
\mathds{Z}_{p^n}^\pi:=\varprojlim\limits_{\pi_k}\mathds{Z}/(p^n)^k\mathds{Z}
\ar[r]^{\quad\qquad\bin_n}&\mathds{F}_p^n[[t]]
}
\]
Here, $\varprojlim\limits_{\pi_k}$ is the projective limit with respect to the inverse system 
$\mathset{\pi_k\mid k\in\mathds{N}}$ (cf.\ \cite[\S 2.5]{Jacobson1989} for the definition of projective limit).

\begin{thm}
The Hilbert curve constructions lead to a continuous surjective map $\mathcal{H}_n$
which fits into the following commutative diagram:
\[
\xymatrix{
\mathds{Z}_{p^n}^\pi\ar[drr]_{\mathcal{H}_n}\ar[r]^{\!\!\!\!\bin_n}
&\mathds{F}_p^n[[t]]\ar[r]^{h_n}
&\mathds{F}_p^n[[t]]\ar[d]^{\coord_n}
\\
&&[0,1]^n
}
\]
The map $\coord_n$ is induced by
\begin{align*}
\coord_n^k&\colon\mathds{F}_p^n[t]/t^k\mathds{F}_p^n[t]\to[0,1]^n
\\
&\sum\limits_{\nu=0}^{k-1}a_\nu t^\nu\mapsto
\sum\limits_{\nu=1}^{k}a_\nu \left(\frac{1}{p}\right)^\nu
\end{align*}
and is
 continuous, the maps $h_n$ and $\bin_n$ are homeomorphisms.
\end{thm}

\begin{proof}
The 
commutative diagram of Theorem \ref{glueHilbert} 
carries over to the case $p>2$ also. The horizontal bijections in that diagram
induce a homeomorphism
\[
h_n\colon \mathds{F}_p^n[[t]]\to\mathds{F}_p^n[[t]]
\]
The commutative diagram
 (\ref{continuous2Cube}) also carries over to the case $p>2$. The horizontal maps there
 induce a  map 
\[
\coord_n\colon\mathds{F}_p[[t]]\to [0,1]
\]
which is surjective, and also continuous, as it is ($1$-)Lipshitz:
\[
\absolute{\coord_n(f)-\coord_n(g)}\le\absolute{f-g}_p
\]
where 
\[
\absolute{x}_p=\begin{cases}
p^{-\nu},&x=t^\nu\cdot u,\;u=\sum\limits_{\mu=0}^\infty\alpha_\mu t^\mu,\;\alpha_0=1\\
0,&x=0
\end{cases}
\]
is the $p$-adic norm.
Thus all assertions are proven.
\end{proof}

Notice that $\coord_n$ is not injective. Namely, 
\begin{align*}
\coord_n\left(\frac{(p-1)t}{1-t}\right)=\sum\limits_{\nu=1}^\infty (p-1)p^{-\nu-1}=\frac{1}{p}=c_n(1)
\end{align*}

\smallskip
We call $\mathcal{H}_n$ an \emph{$n$-dimensional $p$-adic reflected Hilbert curve}.

\smallskip
Notice also that $\absolute{\cdot}_p$ is a non-Archimedean (also called ultrametric) norm on $\mathds{F}_p^n$:
it satisfies the strict triangle inequality
\[
\absolute{x+y}_p\le\max\mathset{\absolute{x}_p,\absolute{y}_p}
\]

%%%%%%%%%%%%%%%%%%%%%%%%%%%%%%%%%%%%%%%%%%%%%%%%%%%%%%%%%%%%%%%
\section{Approximating points on $2$-adic affine Gray-Hilbert curves}
In this section, we will show that, analogous to Theorem \ref{p-adicHC}, the local pieces of $2$-adic affine Gray-Hilbert curves are $2$-adic approximations of that curve.

\smallskip
Let $p=2$ and
\begin{align*}
\epsilon^{(e,d)}&\colon\mathds{Z}/2^n\mathds{Z}\to\mathds{F}_2^n
\\
i&\mapsto\begin{cases}
\gamma_n^{(e,d)}(0),&i=0\\
\gamma_n^{(e,d)}(i-1),&i\equiv 1\mod 2\\
\gamma_n^{(e,d)}(i-2),&i\equiv 0\mod 2
\end{cases}
\end{align*}
and
\begin{align*}
\delta^{(e,d)}&\colon\mathds{Z}/2^n\mathds{Z}\to\mathds{Z}/n\mathds{Z}
\\
i&\mapsto\begin{cases}
\tau_n^{(e,d)}(0),&i=0\\
\tau_n^{(e,d)}(i),&i\equiv 1\mod 2\\
\tau_n^{(e,d)}(i-1),&i\equiv 0\mod 2
\end{cases}
\end{align*}
with
$\tau_n^{(e,d)}$ as in (\ref{tau}).

\begin{lem}\label{direction+entry}
The quantities $\epsilon^{(e,d)}$ and $\delta^{(e,d)}$ define the sequences of intra-sub-hypercube directions and of entry points for $\gc_n^{(e,d)}$.
\end{lem}

\begin{proof}
This is an immediate consequence of \cite[Thm. 2.9, Thm. 2.10]{Hamilton2006}.
\end{proof}

\smallskip
For $k\ge 0$ on smallest local pieces:
\[
\xymatrix{
\gamma_n^{(\epsilon_k,\delta_k)}(i_k)+\mathds{F}_2^n t\ar[r]^{h^{k+1}}\ar[d]_{D_t}
&\gamma_n^{(\epsilon_k,\delta_k)}(i_k)+\mathds{F}_2^n t\ar[d]^{D_t}\\
\mathds{F}_2^n\ar[r]_{\gc_n^{(\epsilon_{k+1},\delta_{k+1})}}&\mathds{F}_2^n
}
\]
with
\begin{align}\label{epsdelta}
\epsilon_{k+1}&=\epsilon^{(\epsilon_k,\delta_k)}(i_k),
\quad\delta_{k+1}=\delta^{(\epsilon_k,\delta_k)}(i_k)
\\
\epsilon_0&=0,\;\delta_0=0\nonumber
\end{align}
Or on $t$-adically expanded local pieces:
\[
\xymatrix{
b_k+\mathds{F}_2^n t^{k+1}\ar[r]^{h^{k+1}}\ar[d]_{D^k_t}
&b_k+\mathds{F}_2^n t^{k+1}\ar[d]^{D_t^k}\\
\mathds{F}_2^n\ar[r]_{\gc_n^{(\epsilon_{k+1},\delta_{k+1})}}&\mathds{F}_2^n
}
\]
with  
\[
b_k=\sum\limits_{\kappa=0}^k\gamma_n^{(\epsilon_\kappa,\delta_\kappa)}(i_\kappa)t^\kappa
\]

We see that
\[
D_t^k=\underbrace{D_t\circ\dots \circ D_t}_{\text{$k$ times}}
\]
with $D_t=D_t^1$.

\begin{lem}
Let $I(t)\in\mathds{Z}/2^n\mathds{Z}[t]$ be such that 
$I\left(2^n\right)$ is the $2^n$-adic expansion of $i\in\mathds{Z}/\left(2^n\right)^{k+1}\mathds{Z}$.
Then 
\[
\pi_k(i)=\Delta_t I(2^n)
\]
where
\begin{align*}
\Delta_t\colon\mathds{Z}/2^n\mathds{Z}[t]&\to\mathds{Z}/2^n\mathds{Z}[t]
\\
\sum\limits_{\kappa=0}^k a_\kappa t^\kappa &\mapsto\sum\limits_{\kappa=1}^k a_\kappa t^{\kappa-1} 
\end{align*}
\end{lem}

\begin{proof}
One readily checks this.
\end{proof}

\smallskip
With the $2^n$-adic expansion
\begin{align}\label{2^n-adic}
i=\sum\limits_{\kappa=0}^k i_k\left(2^n\right)^{k-\kappa}
\end{align}
we have
\begin{align}\label{ed}
e_k(b)&=\epsilon_k,\quad d_k(b)=\delta_k
\end{align}
for 
\begin{align}\label{b(i)}
b=b(i)=\sum\limits_{\kappa=0}^k\gamma_n^{(\epsilon_\kappa,\delta_\kappa)}(i_\kappa)t^\kappa\in\mathds{F}_2^n[t]/t^{k+1}\mathds{F}_2^n[t]
\end{align}

\begin{thm}\label{hilbertmapping}
It holds true that
\[
h_n(x)\equiv h^{k+1}(x_k)=\sum\limits_{\kappa=0}^k\gc_n^{(\epsilon_\kappa,\delta_\kappa)}(a_\kappa) t^\kappa
\mod t^{k+1}
\]
for $x=\sum\limits_{\kappa=0}^\infty a_\kappa t^\kappa\in\mathds{F}_2^n[[t]]$ and
$x_k=\sum\limits_{\kappa=0}^k a_\kappa t^\kappa$, and with
$\epsilon_\kappa$ and $\delta_\kappa$ as in (\ref{epsdelta}), where $i_\kappa=\bin_n^{-1}(a_\kappa)$.
\end{thm}

\begin{proof}
Observe that, by construction,
\[
h^{k+1}(x_k)=h^{k+1}(\bin_n^{k+1}(i))=b(i)
\]
where $i=\bin_n^{-1}(x_k)$ and $b(i)$ as in (\ref{b(i)}) for the $2^n$-adic expansion of $i$ given in
(\ref{2^n-adic}). According to Lemma \ref{direction+entry}, (\ref{ed}) defines the correct intra sub-hypercube directions and entry points.
Clearly, we have
\[
h_n(x)\equiv h^{k+1}(x_k)\mod t^{k+1}
\]
This proves the assertion.
\end{proof}

%%%%%%%%%%%%%%%%%%%%%%%%%%%%%%%%%%%%%%%%%%%%%%%%%%%%%%%%%%%%%%%%%%%%%%%%%%%
\section{Indexing a sample of hypercube points using $p$-adic Hilbert curves}

\subsection{The $p$-adic Gray-Hilbert tree}

Let $\mathcal{T}_n$ be the finite rooted tree in which each non-leaf vertex has precisely $p$ children, and with $p^n$ leaves. 
A \emph{level} is the length of a path from root to a vertex of $\mathcal{T}_n$.
The \emph{$p$-adic Gray-Hilbert tree} $\mathcal{G}_n$ is obtained by gluing a copy of $\mathcal{T}_n$ to each leaf by identifying that leaf with the root of $\mathcal{T}_n$. This copy will be called a \emph{local piece}. 
Each level of a local piece is labelled by the coordinate given by the local ordering of coordinates which
is determined by the transformation used for building up the
Gray-Hilbert curve at this particular hypercube subdivision.
The \emph{$k$-th iteration} of $\mathcal{G}_n$ is the finite subtree having the same root as $\mathcal{G}_n$ and $p^{nk}$ leaves.
A path in the $p$-adic Gray-Hilbert tree from root to a given vertex
corresponds to a sequence of subdivisions of certain coordinates
which yields a subcuboid of $I^n$ whose edges are one of two
kinds: either short or long.

\smallskip
Let $S \subset I^n$ be a finite point cloud. The \emph{Gray-Hilbert tree $\mathcal{G}_n(S)$
generated by $S$} is defined as being the smallest subtree of the
$p$-adic Gray-Hilbert tree $\mathcal{G}_n$ having the same root, and  whose leaves correspond to subcuboids of $I^n$ each containing
precisely one element of $S$. Later in this section, we will show how to obtain $\mathcal{G}_n(S)$ efficiently.

\subsection{The inverse Hilbert curve map: $p=2$.}

According to Theorem \ref{hilbertmapping}, 
the inverse map $\left(h_n^{k+1}\right)^{-1}$ is given for each coefficient by
\[
\left(\gc_n^{(\epsilon,\delta)}\right)^{-1}=\gc_n^{-1}\circ T_{(\epsilon,\delta)}
\]
and $\gc_n^{-1}$ is the map
\[
\mathds{F}_2^n\to\mathds{F}_2^n,\;x\mapsto A\cdot x
\]
with $A=(\alpha_{ij})$ and
\[
\alpha_{ij}=\begin{cases}
1,&i\le j\\
0,&\text{otherwise}
\end{cases}
\]
as can be readily checked.

\subsection{The inverse Hilbert curve map: $p>2$}
In the case $p>2$, the inverse map $\left(h_n^{k+1}\right)^{-1}$ is given coefficient-wise by
\[
\left(\gc_n^{\epsilon_k}\right)^{-1}
=\gc_n^{-1}\circ T_{\epsilon_k}
\]

Let $x\in\mathds{F}_p^n$. We can write
\[
x=(x_{n-1}\;x^{n-1})
\]
with $x^{n-1}=(x_{n-2},\dots,x_0)\in\mathds{F}_p^{n-1}$.

\begin{lem}\label{gcperp}
It holds true that
\[
\gc_n(x)^\perp=\gc_n(x^\perp)
\]
for $x\in\mathds{F}_p^n$.
\end{lem}

\begin{proof}
Here, we will write $d_n$ for the vector $d\in\mathds{F}_p^n$.
We proceed by induction. For $n=1$, the assertion is clear.
Assume that the assertion holds true for $n-1$. Now, we have
\begin{align*}
\gc_n&(x)+\gc_n(x^\perp)=
\left(x_{n-1}\;\gc_{n-1}\left(x^{n-1}\right)^{\sigma_{x_{n-1}}}\right)
\\
&+\left((x^\perp)_{n-1}\;\left(\gc_{n-1}\left(x^\perp\right)^{n-1}\right)^{\sigma_{x_{n-1}}}\right)\\
&=\left(p-1\;\gc_{n-1}\left(x^{n-1}\right)^{\sigma_{x_{n-1}}}\right.
\\
&\quad +\left.\gc_{n-1}\left(\left(x^\perp\right)^{n-1}\right)^{\sigma_{x_{n-1}}}\right)
\\
&=\left(p-1\;w^{\sigma_{x_{n-1}}}\right)
\end{align*}
with
\begin{align*}
w&=\gc_{n-1}\left(x^{n-1}\right)
+\gc_{n-1}\left(\left(x^\perp\right)^{n-1}\right)
\\
&=\gc_{n-1}\left(x^{n-1}\right)+\gc_{n-1}\left(\left(x^{n-1}\right)^\perp\right)
\\
&=d_{n-1}
\end{align*}
as clearly, $\left(x^\perp\right)^{n-1}=\left(x^{n-1}\right)^\perp$.
This proves the assertion.
\end{proof}

\begin{lem}
It holds true that
\[
\gc_n\circ\gc_n=\id
\]
i.e.\ $\gc_n$ is its own inverse.
\end{lem}

\begin{proof}
We proceed by induction. For $n=1$, the assertion is clear.
Assume that the assertion holds true for $n-1$. Now,
we have
\begin{align*}
\gc_n(\gc_n(x))
&=\left(x_n\;\gc_{n-1}\left(\gc_{n}\left(x\right)^{n-1}\right)^{\sigma_{x_{n-1}}}\right)
\\
\gc_n(x)^{n-1}&=\gc_{n-1}\left(x^{n-1}\right)^{\sigma_{x_{n-1}}}
\end{align*}
If $x_{n-1}\equiv 0\mod 2$, then
\begin{align*}
\gc_{n-1}&\left(\gc_n(x)^{n-1}\right)^{\sigma_{x_{n-1}}}
=\gc_{n-1}\left(\gc_{n-1}(x^{n-1})\right)
\\
&=x^{n-1}
\end{align*}
If $x_{n-1}\equiv 1\mod 2$, then
\begin{align*}
\gc_{n-1}&\left(\gc_n\left(x\right)^{n-1}\right)^{\sigma_{x_{n-1}}}
=\gc_{n-1}\left(\gc_{n-1}\left(x^{n-1}\right)^\sigma\right)^\sigma
\\
&\stackrel{\ref{sigmaperp}}{=}\gc_{n-1}\left(\gc_{n-1}\left(\left(x^{n-1}\right)^\perp\right)^\perp\right)
\\
&\stackrel{\ref{gcperp}}{=}\gc_{n-1}\left(\gc_{n-1}\left(x^{n-1}\right)^\perp\right)^\perp
\\
&\stackrel{\ref{gcperp}}{=}\left(\gc_{n-1}\left(\gc_{n-1}\left(x^{n-1}\right)\right)^\perp\right)^\perp
\\
&=\gc_{n-1}\left(\gc_{n-1}\left(x^{n-1}\right)\right)
=x^{n-1}
\end{align*}
From this, the assertion follows.
\end{proof}

%%%%%%%%%%%%%%%%%%%%%%%%%%%%%%%%%%%%%%%%%%%%%%%
\subsection{The indexing method}
The indexing method uses the inverse Hilbert curve map $\left(h^k\right)^{-1}$.

\smallskip
Let $S\subset\mathds{F}_p^n[[t]]$ be a finite subset.

\begin{alg}[Indexing]\label{hilbertindex}
Let $S^0=S$ and $T_0=\emptyset$.

\smallskip
Let for $k\ge 1$
\begin{align*}
S^k&=S^{k-1}\setminus\left\{s\in S^{k-1}\mid \text{$s$ is unique in $S^{k-1}$}\right.
\\
&\left.\text{with $s\equiv s'\mod t^k$ for $s'\in S^{k-1}\mod t^k$}\right\}
\\
S_k&=\left(S^{k-1}\setminus S^k\right)\mod t^{k}
\end{align*}
The map $\left(h^{k}\right)^{-1}$ induces a linear ordering on $S_k$. 

\smallskip
Map $S_k$ via $\rho_{k-1}$ to $\mathds{F}_p^n[t]/t^{k-1}\mathds{F}_2^n[t]$. Then $\left(h^{k-1}\right)^{-1}$ induces a linear ordering on $\rho_{k}(S_k)\cup S_{k-1}$. This induces a linear ordering on
$S_k\dot{\cup} S_{k-1}$, and thus on
\[
T_k=T_{k-1}\cup S_k
\]
Hence, obtain a partitioned set
\[
T_N=S_1\dot{\cup}\dots\dot{\cup} S_N
\]
consisting of partial sums of elements of $S$ such that $T_N$ and $S$ have the same cardinality, and
a linear ordering on $T_N$. This induces a linear ordering on $S$.
\end{alg}

The set $T_N$ coming out of Algorithm \ref{hilbertindex} corresponds to a set of hypercubes of various sizes, each being the largest containing precisely one element of the finite set $S$. The following algorithm
is a simplification of the algorithm in \cite[\S 3.2.2]{Mumfdendro} and
finds from this the Gray-Hilbert tree generated by $S$.

\begin{alg}[Gray-Hilbert tree of a set]\label{GHT(S)}
Let $S\subset\mathds{F}_p^n[[t]]$ be a finite set. 
For a subset $S'\subset S$ let 
\[
T'=\mathcal{G}_n(S')
\]
whose edges are labelled with the digits of the coefficients in the formal power series representing the elements of $S'$.
Start with $s_0\in S$. Set $S_0=\mathset{s_0}$. Then $T_0=\mathcal{G}_n(S_0)$ is the root and a leaf containing $s_0$. 

\smallskip
Assume $s_N\in S\setminus S_{N-1}\neq\emptyset$. In order to create
\[
T_N=\mathcal{G}_n(S_N)
\]
with $S_N=S_{N-1}\cup\mathset{s_N}$, let $v_N$ be the nearest vertex of $T_{N-1}$ on the path from root to the end $s_N$ in the infinite $p$-adic Gray-Hilbert tree.
This vertex can be found by tracing the coefficients in the formal power series representing $s$ until a coefficient is found which has a digit not being a label of an edge emanating from a vertex in $T_{N-1}$. The next vertex is a leaf containing $s_N$.

\smallskip
Continue as in the previous step until $S_N=S$.
\end{alg}

The insert algorithm is as follows:
\begin{alg}[Insert]\label{insert}
Given a finite subset $S\subset\mathds{F}_p^n[[t]]$ and a partition
\[
S=S_1\dot{\cup}\dots\dot{\cup}S_N
\]
as obtained from Algorithm \ref{hilbertindex}, and some $f\in\mathds{F}_p^n[[t]]$.
Then find all $s\in S$ such that 
\[
\absolute{f-s}_p
\]
is smallest possible. This means
\[
f\equiv s\mod t^k,\quad f\not\equiv s\mod t^{k+1}
\]
and $k$ is largest possible in $\mathset{1,\dots,N}$.
If $s$ is unique, then for all $s'\in S\setminus\mathset{s}$:
\[
\absolute{s-s'}_p=\absolute{s-f+f-s'}_p\stackrel{(*)}{=}\absolute{f-s'}_p=p^{-\ell}
\]
with $\ell<k$, where $(*)$ holds true because of the non-archimedean property \cite[Ch.\ 2]{DR2016}. Hence, $s\mod t^j\in S_j$ with $j\le k$. 
So,
update $S_{k+1}$ and $S_j$:
\begin{align*}
S_{k+1}&\leftarrow S_{k+1}\cup\mathset{f,s}\mod t^{k+1}
\\
S_j&\leftarrow S_j\setminus\mathset{s\mod j}
\end{align*}
Now, update the linear ordering with the help of $\rho_k$ as in Algorithm \ref{hilbertindex}.

\smallskip
If $s$ is not unique, then for $s'\neq s$ another such nearest neighbour of $f$ it holds true that:
\[
\absolute{s-s'}_p=\absolute{s-f+f-s'}_p\le p^{-k}
\]
Hence, $s\in S_j$ with $j>k$. So, update $S_{k+1}$:
\[
S_{k+1}\leftarrow S_{k+1}\cup\mathset{f\mod t^k}
\]
and update the linear ordering again with $\rho_k$ as in Algorithm \ref{hilbertindex}.
\end{alg}

%%%%%%%%%%%%%%%%%%%%
\subsection{Upper complexity bounds}

\paragraph{Complexity of finding a point on $h^{k+1}$.} The time complexity of computing $h^{k+1}(x_k)$ equals $k$ times the time complexity of computing $\gc_n^{(e,d)}$ for $n$-bit binary numbers.
As $\gc_n^{(e,d)}$ consists of a bounded number of bit-wise additions, shifts and swaps, independent of $n$,
its time complexity is in $O(n)$. Hence,
the time complexity of $h^{k+1}(x_k)$ is in $O(kn)$.

\paragraph{Complexity of  indexing algorithm.}
For the indexing method (Algorithm \ref{hilbertindex}), the inverse Gray code $\gc_n^{-1}$ needs to be computed.
This  is $n$ times the complexity of $n$-bit operations, so it is in $O(n^2)$.
So, the time complexity of finding the image of a point under $\left(h^{k}\right)^{-1}$ is in $O(kn^2)$.
If the points of the finite set $S\subset\mathds{F}_2^n[[t]]$ are all distinguishable modulo $t^{k+1}$,
then in the worst case, $k$ iterations are needed. As all points of $S$ need to be taken into account, it follows that the overall time complexity is in 
$O(\absolute{S} k n^2)$.

\paragraph{Complexity of Gray-Hilbert tree of a set.} For computing the Gray-Hilbert tree of a finite set $S\subset\mathds{F}_p^n[[t]]$,
all elements of $S$ need to be searched, and in the worst case, $\mathcal{G}_n(S)$ is a finite $p$-regular tree (with $S$ leaf nodes). This gives an upper complexity bound of
$O(\absolute{S}\log\absolute{S})$.

\paragraph{Complexity of Insert.}
Also here, the inverse Gray code needs to be computed, and all elements of $S\subset\mathds{F}_2^n$ need to be considered. This means that the time complexity of Algorithm \ref{insert} is in 
$O(\absolute{S} k n^2)$, as also here, an image under $\left(h^{k}\right)^{-1}$ is needed.

%%%%%%%%%%%%%%%%%%%%%%%%%%%%%%%%%%%%%%
\section{An implementation of a dynamic $2$-adic scaled Hilbert indexing method}\label{implementation}

A binary tree has been chosen as basic data structure for the implementation of a dynamic $2$-adic scaled Hilbert index. The Hilbert value is obtained by interpreting the movement to the first sub-node of a node as false/zero and the movement to the second sub-node as true/one. 
A binary number is built up bit by bit with every step up the tree. This binary number serves us as the Hilbert value of the final leaf node.

\smallskip
Starting from root node with $n$ as the dimension, $n$ steps are needed to fulfil one complete Hilbert curve iteration process. 
The standard Hilbert curve index deals with a constant number of Hilbert curve iterations $k$. 
Therefore, every leaf node of the resulting binary tree is on the same level. The tree will contain $\ell=k \cdot n$ levels and $2^\ell$ leaf nodes. 
The hyper-cuboids of the leaf nodes are equal in size and the tree will have under- and overfilled leaf nodes, if nodes are viewed as `buckets' which may contain up to a certain amount $s$ of data points; if there are more than $s$ points in a `bucket', then the vertex splits, unless the prescribed maximal number of levels is reached. So, if the data is clustered in a certain way, then there will be leaf nodes containing more than $s$ points, while other leaf nodes contain fewer or even none of the points.

\smallskip
The tree covers a maximal $n$-dimensional hyper-cuboid. Each indexed object has a representation as an object in an $n$-dimensional space. In order to insert, delete or retrieve objects, 
the tree must be queried from root to leaf node. One step along a branch splits the hyper-cuboid of a node into two parts by cutting orthogonally 
to a splitting coordinate through the centre of the node's hyper-cuboid. 
For each step up the tree, it has to be decided which splitting coordinate should be taken and how to map the two hyper-cuboid parts to the two sub-nodes. 
There are two cases for how to map the bit of the Hilbert value to the hyper-cuboid part it belongs to. The first case maps the lower part of the hyper-cuboid 
to the first sub-node (Algorithm \ref{split}, line 8) and the upper part of the hyper-cuboid to the second sub-node (Algorithm \ref{split}, line 9). 
The second case does the same vice versa (Algorithm \ref{split}, lines 5 and 6). The Gray code defines which cases have to be chosen within one Hilbert curve iteration process.

\smallskip
How to split a hyper-cuboid within one Hilbert curve iteration process also depends on the path taken within the previous Hilbert curve iteration process. 
To connect the ends of each sub-hyper-cuboids with the start of the next neighbouring sub-hyper-cuboid on the next Hilbert curve iteration correctly, 
it is necessary to know if the start or end should lie in the lower or upper part of the sub-hyper-cuboid regarding each coordinate. 
It is also important which splitting coordinate comes first in order to rotate the sub-hyper-cuboid correctly. 
The order of the following coordinates does not matter when building a space filling curve (Corollary \ref{numberOfGC}). 

\smallskip
The following implementation uses a list of coordinate identifiers and pulls up the first coordinate $d$ to the front of the list for the next Hilbert curve iteration process using the following permutation:
\[
\begin{pmatrix}
n-1&n-2&\dots &d&d-1&\dots&1&0\\
d&n-1&\dots&d+1&d-1&\dots&1&0
\end{pmatrix}
\]
and call this Gray-Hilbert curve \emph{bubble}.

We also use a ring of coordinate identifiers and point to the coordinate $d$ to start with at each start of a Hilbert curve iteration process, using the following permutation:
\[
\begin{pmatrix}
n-1&n-2&\dots &n-(d+1)&\dots&0\\
d&d-1&\dots&0&\dots&d+1
\end{pmatrix}
\]
and call this curve \emph{ring}.

\smallskip
The developed dynamic variant of the Hilbert index deals with a maximal number of objects $s$ contained within each leaf node. 
It dynamically splits a leaf node if the maximal number is exceeded, or collapses a leaf node if no objects are left within the leaf node.
But if the prescribed maximal number of levels is reached, no splitting occurs, and then there may be \emph{overfilled} leaves containing more than
$s$ objects.

\smallskip
This mechanism takes place on any level of the tree. Therefore, each object is not indexed on a certain Hilbert curve iteration level. 
%All levels, and therefore each sub-level or sub-hyper-cuboid part, respectively any coordinate within a Hilbert curve iteration process $i$ can be a leaf node.

\begin{alg}\label{split}
\emph{Input}. A tree $T$ covering an $n$-dimensional hyper-cube $C$ and a point $x\in C$

\smallskip\noindent
\emph{Output}. The leaf node $\nu$ where $x$ lies in sub-hypercube $c$.

\smallskip\noindent
\emph{Pseudo-code} of the routine FINDNODE($T$,$C$,$x$):
\begin{lstlisting}
01. nu:=root(T), c:=C, h:=false, 
     S:= create a list of coord id's, 
     a := S.start, a' := S.end, 
     a'' := S.end, e := null
02. while nu is not a leaf
03.  if h
04.  then {if x in lower part of c 
                            resp. a}
05.   then {nu:=nu.right, a':=a, 
            c:=lower part of c resp. a}
06.   else {nu:=nu.left, a'':=a, 
        c:=upper part of c resp. a, 
        h:=false}
07.  else {if x in lower part of c 
            resp. a}
08.   then {nu:=nu.left, a'':=a, 
        c:=lower part of c resp. a}
09.   else {nu:=nu.right, a':=a, 
        c:=upper part of c resp. a, 
        h:=true}
10.  if S.hasNext
11.  then a:=S.next, if a equals e
12.   then {if x was lower part of c 
        resp. a then h:=-h}
13.   else {if x was upper part of c 
        resp. a then h:=-h}
14.  else if h
15.   then S.start:=a'' 
            % pull a'' to the start of S
16.   else S.start:=a'  
            % pull a' to the start of S
17.   a:=S.start, a':=S.end, 
        a'':=S.end, e:=S.end
18.   if {x was in upper part of c 
        resp. a and v moved to v.left 
            resp. a} 
      or {x was in lower part of c 
        and v moved to v.right 
            resp. a}
19.   then h:=true
20.   else h:=false
21. end while
22. return nu 
\end{lstlisting}
\end{alg}

%%%%%%%%%%%%%%%%%%%%%%%%%%%%%%%%%%%%%%%%%%
\section{Storage efficiency of the $p$-adic scaled Gray-Hilbert tree}

%\smallskip
It is needless to say that the $p$-adic scaled Gray-Hilbert tree for a data set $S\subset\mathds{R}^n$ is storage efficient, because it is the smallest sub-tree of the Gray-Hilbert tree containing $S$ as its leaf nodes.
The aim of this section is to compare the storage properties of the scaled Gray-Hilbert tree with the static $p$-adic Gray-Hilbert index.

\smallskip
Given some multi-dimensional point cloud $S$,
the nodes of the Gray-Hilbert tree are viewed as \emph{buckets} which are either \emph{empty}, \emph{filled}, \emph{underfilled}, or \emph{overfilled}, 
respectively, depending on if the corresponding hypercuboid contains no, a specified number, less than, or more than a pre-specified number $s$ of data points from $S$. This number is 
called \emph{bucket capacity}, and we assume that $s<\absolute{S}$, where $\absolute{S}$ denotes the cardinality of $S$.

\smallskip
We define the following quantities
\begin{align*}
\Omega(T,S)&=(1+\omega(T,S))
\times\absolute{L(T)}
\\
\omega(T,S)&=
\frac{\text{$\#$ overfilled leaf nodes of $T$}}
{\text{$\#$ non-empty leaf nodes of $T$}}
\end{align*}
where  
\[
L(T)=\mathset{\text{leaf nodes of $T$}}
\]
and $T$ is a finite sub-tree of the Gray-Hilbert tree $\mathcal{G}_n$.

\smallskip
The quantity $\Omega(T,S)$ is a measure incorporating the capacity of $T$ (number of leaf nodes) and a factor 
$1+\omega(T,S)$ where $\omega(T,S)$ is the number of leaf node splittings which would occur for $S$, if the depth of the tree $T$ were not bounded. Multiplying with this factor means virtually creating space for data which everywhere would have the same overfilling property as $S$. In a sense, this would mean to transfer the average local density of $S$ within the occupied part of the space to the whole hyper-cube.
We call $\Omega(T,S)$ the \emph{capacity of $T$ for $S$}.

\smallskip
Given two trees $T_1$ and $T_2$ for indexing $S$, the \emph{capacity ratio}
is defined as
\[
R(T_1,T_2;S)=\frac{\Omega(T_1,S)}{\Omega(T_2,S)}
\]
This will be applied to the scaled and an optimal static index for $S$.

\smallskip
We will write
\[
T_{\scaled}=\mathcal{G}_n(S)
\]
for the Gray-Hilbert tree generated by $S$. The tree
\[
T_{\static}
\]
is defined as the finite Gray-Hilbert tree with maximal iteration number
\[
k=\left\lceil\frac{\log_p\dfrac{\absolute{S}}{s}}{n}\right\rceil
\]
This is motivated by wanting to have
\[
p^{nk}\approx \frac{\absolute{S}}{s}
\]
as the number of leaf nodes in $T_{\static}$. This is the smallest sub-tree of the $p$-adic Gray-Hilbert tree such that on average no leaf node is overfilled.

\smallskip
In order to compare the scaled and the static Gray-Hilbert curve indices, we consider the ratio
\[
R_p(S)=R(T_{\scaled},T_{\static};S)=\frac{\Omega(T_{\scaled},S)}{\Omega(T_{\static},S)}
\]
which we call the \emph{$p$-adic Gray-Hilbert capacity ratio for $S$}.
If $R_p(S)<1$, then the scaled index is more storage efficient than the static index, otherwise not.

\smallskip
In order to understand the behaviour of $R_p(S)$ for increasing dimension, first consider the maximal iteration number $k$, which can be written as
\[
k=\frac{\log_p\dfrac{\absolute{S}}{s}}{n}+\epsilon_p
\]
with $\epsilon_n\in [0,1)$. We have
\begin{align}
\frac12 p^{-n\epsilon_p}&=
\frac{\absolute{S}/s}{2p^{nk}}\le
R_p(S)\le\frac{\absolute{S}-s+1}{p^{nk}}\label{pcapacity}
\\
&\le s\cdot\frac{\absolute{S}/s}{p^{nk}}
=s p^{-n\epsilon_p}\nonumber
\end{align}
because
\[
\absolute{S}/s\le\absolute{L(T_{\scaled})}
\le \absolute{S}-s+1
\]
From this, it follows that 

\begin{align}
R_p(S)\to 0
\end{align}
for $n\to\infty$, as 
\[
\lim\limits_{n\to\infty}\epsilon_p= 1
\]
In particular, if the dimension $n$ is sufficiently high, then $R_p(S)\le 1$.
Namely, if
\begin{align}\label{criterion}
\epsilon_p\ge \frac{\log_p s}{n}
\end{align}
then this is the case. So, if (\ref{criterion}) is satisfied, then
the scaled $p$-adic Gray-Hilbert index index performs better then the best static $p$-adic Gray-Hilbert index, the latter being optimal in the sense that on average, no leaf node is overfilled. We summarise this as follows:

\begin{thm}
If $\frac{\log_p s}{n}<\epsilon_p$ holds true, then $T_{\scaled}$ is 
more storage efficient than $T_{\static}$, and that relative efficiency increases with increasing dimension $n$. For fixed $n$, this is the case, if the bucket capacity $s$ is sufficiently small.
\end{thm}

The inequalities (\ref{pcapacity}) are equivalent to
\[
0\le-\log_p R_p(S)-n\epsilon_p+\log_p s\le \log_p (2s)
\]
so that in order to have a measure between zero and one,
we define the quantity
\begin{align*}
\rho_p(S,s)&=\frac{-\log_p(R_p(S)/s)-n\epsilon_p}
{\log_p (2s)}\in[0,1]
\end{align*}
We call this the \emph{$p$-adic Gray-Hilbert local sparsity measure for $S$ and $s$}. It holds true that
\[
(2s)^{\rho_p(S,s)}=\frac{\absolute{S}}
{\absolute{L(T_{\scaled})}}
\left(1+\omega(T_{\static})\right)
\]
Observe that  the dependence on $p$ is through $T_{\static}$ and $T_{\scaled}$,
while the dependence on $s$ is through the definition of a leaf being overfilled in $\omega(T_{\static},S)$ and $L(T_{\scaled})$.
This means that
for fixed bucket capacity $s$, the quantity $\rho_p(S,s)$ should depend on the distribution of $S$ only, the question being in which way.

\smallskip
In any case, a value of $\rho_p(S,s)$ near one means that very many leaf nodes in the $T_{\static}$ are overfilled, possibly 
because of local non-sparsity, whereas in the case of $\rho_p(S,s)$ near zero, only few leaf nodes of $T_{\static}$ are overfilled, possibly because of a higher local sparsity of the distribution of $S$.

\smallskip
Experiments were done in the $2$-adic case with various different data sets. Namely, the iris data set \cite{Fisher1936} in four different encodings, as well as randomly generated data.

\smallskip
The iris data set consists of 150 points.
The four different encodings of this data set are listed in Table \ref{irisEncodings}. The complete disjunctive form (cdf) and BCH 127 50 27 contain only $0$ and $1$ values. The latter is an error-correcting
code named after Hocquenghem, Bose and Ray-Chaudhuri \cite{Hocquenghem1959, BR1960}. This encoding of the iris data set was also used in \cite{Brad2017}, whereas its complete disjunctive form 
was studied in \cite{Murtagh2007}. The data set named \emph{cdf trimmed} is obtained by projecting out all coordinates containing only $0$ values.

\smallskip
The random point clouds are 
$100,000$ samples taken from $\mathds{R}^3$ uniformly and under a standard normal distribution, repeated 100 times.

\begin{table}
\centering
\begin{tabular}{|l|r|}
\hline
\rule{0pt}{3mm}endoding&dim\\\hline
original&4\\
complete disjunctive form (cdf)&400\\
cdf trimmed & 123\\
BCH 127 50 27&431\\\hline
\end{tabular}
\caption{Different encodings of the iris data set and their dimensions.}\label{irisEncodings}
\end{table}

\bigskip
Table \ref{rho_iris} shows the Gray-Hilbert local sparsity measure for the case $p=2$ with
the different encodings of the iris data set using \emph{bubble} and \emph{ring} (as defined in Section \ref{implementation}). Both methods yield the same values. It was observed that the high-dimensional encodings yield no overfilled leaf nodes in $T_{\static}$. This can be seen as an indicator that in these cases the data could be locally sparse, which agrees well with the ultrametricity property found in \cite{Murtagh2007}.
The case $s=1$ reflects the fact that $\omega(T_{\static})=0$ in that case, as then the number of leaf nodes of $T_{\scaled}$ equals the number of points in $S$.

\begin{table}
\centering
\begin{tabular}{|c|cc|c|c|c|}\hline
$s$&\multicolumn{2}{|c|}{original}&	{cdf trimmed}&	{cdf}&	{BCH}\\
&\emph{bubble}&\emph{ring}&&&\\\hline 
1&	0.79&0.79&	0.00&	0.00	&0.00	\\
2&	0.61&0.63&	0.22&0.22	&0.22\\
4&	0.59&0.58&	0.36	&0.36&	0.36\\
8&	0.59&0.58&	0.52	&0.52&	0.52\\
16&	0.68&0.68&	0.48&	0.48&	0.48\\
32&	0.72&0.72&0.44&	0.44&	0.44\\
64&	0.80&0.80&	0.42&	0.42&	0.42\\\hline
\end{tabular}
\caption{$2$-adic Gray-Hilbert local sparsity measures  for different encodings of the iris data set using the methods \emph{bubble} and \emph{ring}. Both methods yield the same values.}\label{rho_iris}
\end{table}

\begin{table}
\centering
\begin{tabular}{|c|cc|c|c|c|}\hline
&\multicolumn{2}{c|}{original}&cdf trimmed&cdf&BCH\rule{0pt}{4mm}\\
s&\emph{bubble}&\emph{ring}&&&\\\hline
1&147&147&147&147&147\rule{0pt}{3.5mm}\\
2&99&97&108&108&108\\
4&56&57&70&70&70\\
8&32&33&35&35&35\\
16&20&20&28&28&28\\
32&9&9&24&24&24\\
64&3&3&19&19&19\\\hline
\end{tabular}
\caption{The number of leaf nodes in $T_{\scaled}$ for different encodings of the iris data set. The methods \emph{bubble} and \emph{ring} yield different results only for the original data set.}\label{iris_leafnodes}
\end{table}

\smallskip
Table \ref{iris_leafnodes} shows that the number of leaf nodes in the $2$-adic scaled tree is higher for the high-dimensional encodings of the iris data set than for the original data, if the bucket capacity $s$ is greater than one. In fact, the trees $T_{\scaled}$ are isomorphic for the three high-dimensional encodings.

\begin{table}
\centering
\begin{tabular}{|c|cc|cc|}\hline
&\multicolumn{2}{c|}{uniform}&\multicolumn{2}{c|}{normal}\rule{0pt}{3.5mm}\\
$\log_2 s$&\emph{bubble}&\emph{ring}&\emph{bubble}&\emph{ring}\\\hline
0&0.24&0.24&0.55&0.55\rule{0pt}{3.5mm}\\
1&0.34&0.34&0.47&0.47\\
2&0.59&0.59&0.70&0.70\\
3&0.62&0.62&0.72&0.72\\
4&0.70&0.70&0.72&0.72\\
5&0.77&0.77&0.82&0.82\\
6&0.80&0.80&0.81&0.81\\
7&0.83&0.83&0.81&0.81\\
8&0.84&0.84&0.86&0.86\\
9&0.86&0.86&0.86&0.86\\
10&0.87&0.87&0.84&0.84\\\hline
\end{tabular}
\caption{$2$-adic Gray-Hilbert local sparsity measures for 100,000 random points in the unit cube from a uniform or standard normal distribution.}\label{random}
\end{table}

\smallskip
Table \ref{random} gives the values of $\rho_2(S,s)$ for the random samples from 100,000 points in the unit cube from a uniform or standard normal distribution. For small bucket capacity values $s$, the uniform distribution has a lower value than the normal distribution, but as $s$ increases, it becomes locally less sparse than the normal distribution.

%%%%%%%%%%%%%%%%%%%%%%%%%%%%%%%%%%%%%%%%%%%%%%%%%%%%%%%%%%
\section{Conclusion}

The $p$-adic generalisation of higher-dimensional Hilbert curves can be accomplished by using the $p$-adic Gray code. Restricting to affine transformations, which are in this case coordinate permutations plus translations, allows the mathematical construction of a plethora of well-defined space-filling curves called $p$-adic affine Gray-Hilbert curves. Each of these can be used for indexing $n$-dimensional point clouds by taking the smallest sub-tree of the tree associated with a $p$-adic Gray-Hilbert curve, whose leaf nodes correspond to hypercuboids containing the point cloud. This is the scaled $p$-adic Gray-Hilbert index. It has efficient access and insert algorithms and can be shown to be more storage efficient than the optimal static version of the $p$-adic Gray-Hilbert index tree, if the dimension is sufficiently high or, equivalently, the bucket capacity is chosen sufficiently low. Optimality is defined here as giving rise to the smallest finite $p$-regular tree containing the data in its leaf nodes in such a way that on average no leaf node is overfilled.
The derived $p$-adic local sparsity measure thus contains information about the local distribution of the data set. This has been tested for different encodings of the famous iris data set with the result that the high-dimensional encodings seem to be locally uniform, which agrees with previous findings about their discovered ultrametricity. Potential applications for the index structures constructed here should be seen in the context of massive data in any finite dimension.

\section*{Funding}
This work is supported by the Deutsche Forschungsgemeinschaft (DFG) [grant numbers BR 3513/12-1, BR 2128/18-1]. 

%%%%%%%%%%%%%%%%%%%%%%%%%%%%%%%%%%%%%%%%%%%%%%%%%%%5
\section*{Acknowledgement}
Norbert Paul is thanked for fruitful discussions.
%%%%%%%%%%%%%%%%%%%%%%%%%%%%%%%%%%%%
\bibliographystyle{plain}
\bibliography{biblio}

\end{document}